\documentclass[journal,twoside,web]{ieeecolor}
\usepackage{generic}
\usepackage{cite}
\usepackage{amsmath,amssymb,amsfonts}
\usepackage{algorithmic}
\usepackage{graphicx}
\usepackage{textcomp}
\usepackage{subcaption}

\usepackage{epsfig} 
\usepackage{cite}

\usepackage{float}
\newtheorem{theorem}{Theorem}[section]
\newtheorem{corollary}{Corollary}[section]
\newtheorem{lemma}{Lemma}[section]

\newtheorem{example}{Example}[section]
\newtheorem{proposition}{Proposition}[section]

\renewcommand{\thesection}{\Roman{section}}

\def\BibTeX{{\rm B\kern-.05em{\sc i\kern-.025em b}\kern-.08em
    T\kern-.1667em\lower.7ex\hbox{E}\kern-.125emX}}
\title{Coordination Games on Multiplex Networks: Consensus, Convergence, and Stability of Opinion Dynamics}

\author{
Ruey-An Shiu$^{a}$ and 
Parinaz Naghizadeh$^{b}$
\thanks{This work is supported in part by the NSF under award CCF-2416311.}
\thanks{$^{a}$ Department of Mathematics, National Taiwan University; \texttt{b10201034@ntu.edu.tw}. $^{b}$  Department of Electrical and Computer Engineering, University of California, San Diego; \texttt{parinaz@ucsd.edu}.}
}

\begin{document}

\maketitle
\thispagestyle{empty}

\begin{abstract}
This paper studies opinion dynamics in multilayer networks. Extending a single-layer model, we formulate opinion updates as a synchronous coordination game in which agents minimize a local cost to stay close to their neighbors' opinions. We propose two coupling mechanisms between layers: (i) a merged model that aggregates layers through weighted influences, and (ii) a switching model that periodically alternates across layers. Using random-walk and spectral analysis, we derive sufficient conditions for consensus, characterize convergence rates, and analyze stability under network perturbations. We show that multilayer interactions can induce or accelerate global consensus even when no single layer achieves it alone, and conversely, that individually coordinated layers may lose consensus once interconnected. A common thread arising from our analysis is that the alignment in the weighted degrees of the nodes between the two layers is the main determinant of whether merging or switching can speed up convergence to consensus compared to layers operating in isolation, providing network design intervention guidelines.\looseness-1
\end{abstract}


\section{Introduction}\label{sec:intro}
Opinion formation in modern societies is shaped by the interplay of influences across multiple communication modes, spanning traditional media and face-to-face discussions to algorithm-driven social-media feeds and private-messaging platforms. Individuals are therefore subject to influence from a \emph{multiplex} landscape, simultaneously or sequentially engaging with diverse sources that shape their beliefs. Understanding the formation and evolution of opinions thus requires considering the multiple modalities through which individuals interact.

The formal framework of \emph{multilayer networks} has been widely used to model and analyze emergent phenomena on networks of interconnected networks \cite{boccaletti2014structure,kivela2014multilayer,aleta2019multilayer}. In particular, when all interacting networks share the same set of nodes, as in the case where the same individuals engage across multiple online and offline social networks, the resulting multilayer network is referred to as a \emph{multiplex network}. Surveys of the field \cite{boccaletti2014structure,kivela2014multilayer,aleta2019multilayer} highlight the empirical prevalence of such multiplex structures in opinion formation and dynamics.

Motivated by these observations, we study opinion dynamics over a two-layer multiplex network, where each layer represents a distinct network; the layers can be interpreted as online vs. in-person interactions, or as different social media platforms. Classical single-layer frameworks such as DeGroot averaging \cite{degroot1974reaching} and its game-theoretic extensions \cite{friedkin1990social, shah2009gossip, ghaderi2014opinion} predict consensus on one network at a time; here, we extend these models to capture situations where opinions are shaped by multiple networks. Specifically, building on the single-layer quadratic coordination game model of \cite{ghaderi2014opinion}, we propose and analyze two potential coupling mechanisms between network layers. In the \emph{merged layers} setting, individuals integrate influences from both layers at every time step, with the combined interaction matrix given by a convex combination of the individual layer matrices. In the \emph{switching layers} setting, individuals alternate between layers over time, interacting on one layer for $k$ consecutive steps before switching to the other layer for one step, and repeating this cycle. These two interaction modes reflect different real-world scenarios: the former captures opinion dynamics under \emph{simultaneous exposure} to multiple information sources, while the latter is chosen to capture \emph{attention shifts} across platforms or contexts. 

The main analytical challenges in the study of our proposed models stem from assessing the spectral properties of the resulting update operators. In the merged model, the effective opinion update matrix is a weighted sum of the layer matrices, whose convergence behavior is not simply inherited from the individual layers. In the switching model, the dynamics involve products of the form $BA^{k}$, where $A$ and $B$ are the transition matrices of the two layers; as shown in prior studies of products of stochastic matrices \cite{hajnal1958weak, anthonisse1977exponential, xia2015products, chevalier2017sets}, such products need not inherit the convergence properties of their factors. 

In this work, we employ spectral tools from random-walk theory \cite{diaconis1991geometric, montenegro2006mathematical, bremaud2013markov, aldous-fill-2014, levin2017markov} and classical perturbation analysis \cite{schweitzer1968perturbation, o1993entrywise, horn2012matrix} to overcome these challenges, and provide three fundamental sets of results for each proposed model: (i) conditions on the spectral properties of the layers' transition matrices when consensus is guaranteed to form in each coupling mode, (ii) explicit convergence bounds on how rapidly consensus emerges, and (iii) sufficient conditions for stability under small perturbations to layer weights or connections. We support our theoretical findings using numerical experiments on a real-world high school contact data. Additional experiments, and the full proofs, are provided in the online appendix~\cite{shiu2026coordination}. 

\paragraph{The Main Takeaways}
In the merged layers model (our model of \emph{simultaneous exposure}), we find that 
exposure to a sufficiently stabilizing source can negate the destabilizing effects of another platform that would not individually support consensus. Even so, the speed of convergence to consensus will still depend on how ``similar'' the layers are. Specifically, we show that the appropriate notion of similarity is whether the nodes have similar (weighted) degrees across layers. 

In the switching layers model (our model of \emph{attention shifts}), consensus depends on whether a full round of alternating interactions manages to circulate opinions throughout the network, which can happen if the attention shifts are not too frequent. As a positive byproduct, we show that attention shifts can sometimes unlock stability even when neither layer converges alone. Here, we again find similarity of \emph{the weighted degrees of nodes} determining the speed of convergence. 

The main common takeaway across both models is that the (speed of) emergence of consensus in multiplex networks is not just about whether each layer is strong on its own, but about \emph{whether the layers align well enough, particularly by having the perceived importance of each agent being relatively similar across layers}. 
This perspective suggests that interventions (e.g., reinforcing links in one platform) can make consensus emerge, but that to accelerate agreement, these interventions should also make the platforms more ``homogeneous'', and so inevitably, shape the final consensus opinion. 

\vspace{0.1in}

\paragraph{Related Work}
Our work is most closely related to two lines of literature: (i) characterization of Nash equilibria of coordination games and consensus in single-layer networks, and (ii) analysis of games on multilayer/multiplex networks.

In the single-layer setting, classical averaging models established consensus conditions in connected graphs via Markov-chain techniques \cite{degroot1974reaching}. Game-theoretic extensions introduced quadratic cost formulations to capture social influence and leveraged spectral properties such as connectivity and the SLEM to ensure equilibrium uniqueness \cite{friedkin1990social, shah2009gossip, acemoglu2011opinion, ghaderi2014opinion, olshevsky2014linear}. Subsequent work has incorporated node heterogeneity, including stubborn agents \cite{ghaderi2013opinion} and degree-weighted influence mechanisms \cite{yildiz2013binary, cheng2025degree}, highlighting how high-degree nodes affect convergence dynamics.

Moving beyond single-layer networks, multilayer and multiplex networks can capture more realistic interaction patterns by modeling simultaneous or periodic engagement across layers \cite{boccaletti2014structure, kivela2014multilayer,aleta2019multilayer}. Game theoretical studies in this domain have investigated the existence and multiplicity of equilibria of general games played over multilayer network structures \cite{boruah2024multilayer, ebrahimi2025united}, as well as evolutionary games in multilayer contexts, demonstrating how interlayer links sustain cooperative behavior \cite{gomez2012evolution, wang2015evolutionary, aleta2019multilayer}. To our knowledge, coordination games (which capture DeGroot opinion dynamics) in multilayer networks have remained unexplored. A recent work in \cite{abedinzadeh2025stability} has explored the stability and robustness of opinion dynamics under a Friedkin–Johnsen model on  temporal multiplex networks; our ``switching layer'' model similarly considers the time-varying aspect of interactions across different network layers and its impact on opinion dynamics and stability, but differs due to adopting the DeGroot model of opinion updates.   

\section{Preliminaries: Single-layer Games}
\label{sec:model_setting}

\emph{a) The model:} Consider a network of $N$ agents as an undirected, weighted graph $G_1 = (\mathcal{V}, \mathcal{E}_1, W_1)$, where $\mathcal{V}$ represents agents, and $\mathcal{E}_1 \subseteq \mathcal{V} \times \mathcal{V}$ denotes pairwise interactions. Each edge $(i, j) \in \mathcal{E}_1$ carries a positive weight $w_{ij}^{(1)} > 0$, indicating the influence of agent $j$ on agent $i$. For non-edges, $w_{ij}^{(1)} = 0$, and self-loops are absent ($w_{ii}^{(1)} = 0$). The set of agents influencing agent $i$ is defined as $\partial_i^{(1)} = \{ j \in \mathcal{V} \mid (i, j) \in \mathcal{E}_1 \}$. 

Each agent $i \in \mathcal{V}$ holds an initial opinion $x_i(0) \in [0, 1]$, reflecting their stance on a topic such as support for an idea or proposed project, with $x_i(0) = 1$ and $x_i(0) = 0$ representing full support and extreme opposition, respectively. The initial opinion vector is $\mathbf{x}(0) = [x_1(0), \ldots, x_N(0)]^T$.

Agents update their opinions via a synchronous coordination game, with each agent $i$ minimizing the cost function
\begin{equation}
J_i(x_i, \mathbf{x}_{\partial_i^{(1)}}) = \frac{1}{2} \sum_{j \in \partial_i^{(1)}} w_{ij}^{(1)} (x_i - x_j)^2, \label{eq:cost_function}
\end{equation}
where $\mathbf{x}_{\partial_i^{(1)}} = \{ x_j \mid j \in \partial_i^{(1)} \}$ is the opinion vector of $i$'s influencers. This cost promotes alignment of $i$'s opinion with its influencers. Any opinion vector $\mathbf{x} = [x_1, \ldots, x_N]^T$ with $x_1 = \cdots = x_N$ (i.e., consensus) is a Nash equilibrium.

\emph{b) Agents' opinion dynamics:} At time $t$, agent $i$ updates their opinion by minimizing the cost function \eqref{eq:cost_function}, with the first-order condition yielding the optimal update rule:
\begin{equation*}
x_i(t+1) = \frac{\sum_{j \in \partial_i^{(1)}} w_{ij}^{(1)} x_j(t)}{\sum_{j \in \partial_i^{(1)}} w_{ij}^{(1)}}. 
\end{equation*}
Inspired by this, a row-stochastic \emph{transition matrix} $A$ is used to capture the evolution of agents' opinion; 
its entries are
\begin{equation}
A_{ij} =
\begin{cases}
\frac{w_{ij}^{(1)}}{\sum_{k \in \partial_i^{(1)}} w_{ik}^{(1)}} & \text{if } (i, j) \in \mathcal{E}_1, \\
0 & \text{otherwise}.
\end{cases} \label{eq:matrix_A}
\end{equation}
Then, the opinion dynamics follows $\mathbf{x}(t+1) = A \mathbf{x}(t)$, so that 
\begin{equation}
\mathbf{x}(t) = A^t \mathbf{x}(0). \label{eq:dynamics}
\end{equation}

\emph{c) Consensus.} Assuming $A$ is primitive 
we can view $A$ as the transition matrix for a random walk on $G_1$. By classical random walk theory~\cite{bremaud2013markov, aldous-fill-2014}, as $t \to \infty$, the matrix $A^t$ converges to $\begin{bmatrix}
\boldsymbol{\pi}^T, 
\ldots, 
\boldsymbol{\pi}^T 
\end{bmatrix}^T$, 
where $\boldsymbol{\pi} = [\pi_1, \ldots, \pi_N]$ is the unique stationary distribution of $A$, with
\begin{align}
    \pi_i = \frac{d^{(1)}_{i}}{2|E_1|}, \label{eq:stationary_dist}
\end{align}
where $d^{(1)}_{i} = \sum_{j=1}^N w_{ij}^{(1)}$ is agent $i$'s weighted degree and $|E_1| = \frac{1}{2} \sum_{i=1}^N d^{(1)}_{i}$ is half the total edge weight.  
Consequently, the consensus opinion at equilibrium is 
\begin{equation}
x_i(\infty) = \sum_{j=1}^N \pi_j x_j(0), \quad \forall i \in \mathcal{V}. \label{eq:consensus}
\end{equation}

\emph{d) Analyzing convergence:} To analyze convergence, we equip $\mathbb{R}^N$ with the scalar product $\langle \mathbf{z}, \mathbf{y} \rangle_\pi = \sum_{i=1}^N z_i y_i \pi_i$, 
and the associated norm $\|\mathbf{z}\|_\pi = \left( \sum_{i=1}^N z_i^2 \pi_i \right)^{1/2}$.
Further, for any stochastic matrix $M$, we denote its eigenvalues by $\lambda_1(M) \geq \lambda_2(M) \geq \cdots \geq \lambda_N(M)$. 
Finally, define the error vector at time $t$ as the difference between the opinion vector at time $t$ and the consensus opinion, i.e., 
$\mathbf{e}(t) := \mathbf{x}(t) - \mathbf{x}(\infty)$. 
The following lemma from \cite{ghaderi2014opinion} establishes that the error term converges geometrically.

\begin{lemma}[\cite{ghaderi2014opinion}, Lemma 2]
\label{lem:convergence_rate}
Under the dynamics \eqref{eq:dynamics}, the error $\mathbf{e}(t) = \mathbf{x}(t) - \mathbf{x}(\infty)$ satisfies
\begin{equation*}
\|\mathbf{e}(t)\|_\pi \leq \rho_2(A)^t \|\mathbf{e}(0)\|_\pi, 
\end{equation*}
where $\rho_2(A) = \max_{i \neq 1} |\lambda_i(A)|$ is the second largest eigenvalue modulus (SLEM) of $A$.
\end{lemma}
In words, the variance of the opinions under the stationary distribution converges exponentially fast, at a rate dominated by the second largest eigenvalue of the averaging matrix $A$. 
The SLEM $\rho_2$ characterizes the underlying graph structure. Generally, graphs with more symmetric structures, such as complete graphs, exhibit smaller $\rho_2$ values, thereby achieving faster mixing processes and accelerating convergence to the stationary distribution.

For subsequent analysis, we further equip both matrices and vectors with the max norm.  Given a matrix \(M=[M_{ij}]\) and a vector
\(\mathbf{v}=(v_{1},\dots,v_{N})\), define
\begin{center}
$
\|M\|_{\max} = \max_{i,j}\lvert M_{ij}\rvert,
\qquad
\|\mathbf{v}\|_{\max} = \max_{1\le i\le N}\lvert v_{i}\rvert.
$
\end{center}
When the matrix $A$ is primitive, its stationary distribution
satisfies $\pi_{j}>0$ for every $1\leq j\leq N$.  Consequently, by Lemma \ref{lem:convergence_rate}, we
obtain the following bound on the error.
\begin{corollary}
Under the dynamics~\eqref{eq:dynamics},
\begin{center}
$
\|\mathbf{e}(t)\|_{\max}
\lesssim
\rho_{2}(A)^{t}\,
\|\mathbf{e}(0)\|_{\pi}.
$
\end{center}
\end{corollary}
Here, $\lesssim$ indicates a bound up to a multiplicative constant, which depends on the graph. In words, the corollary states that the \emph{largest} deviation from consensus decays exponentially. 

\emph{e) Stability of the Stationary Distribution:} 
The following classical result on the stability of stationary distributions under matrix perturbations shows how the stationary distribution changes when the transition matrix is perturbed.

\begin{theorem}[\cite{schweitzer1968perturbation}, Theorem 2]
\label{thm:schweitzer}
Suppose $P$ and $\tilde{P}$ are $N \times N$ stochastic matrices with unique stationary distributions $\boldsymbol{\pi}$ and $\boldsymbol{\tilde{\pi}}$, respectively. For the perturbation $E = \tilde{P} - P$, 
\begin{equation*}
\boldsymbol{\tilde{\pi}} - \boldsymbol{\pi} = \boldsymbol{\pi} E Z,
\end{equation*}
where $Z = (I - P + \mathbf{1}\boldsymbol{\pi}^T)^{-1}$.
\end{theorem}
As each entry of the stationary distribution satisfies $0 \leq \pi_i \leq 1$, 
we immediately obtain the following corollary:
\begin{corollary}
\label{cor:stability_bound}
\[
\|\boldsymbol{\tilde{\pi}} - \boldsymbol{\pi} \|_{\max} \lesssim  \|\tilde{P} - P \|_{\max}.
\]
\end{corollary} 
\section{Merged Layers Opinion Dynamics}
\label{sec:multilayer_extension}

\subsection{Merged Layers Coordination Game Model}
\label{subsec:merged_multilayer}

We now proceed to proposing our first extension of the classical single-layer coordination games. 
To model the \emph{simultaneous}  influence of multiple interaction modalities, such as agents being influenced by both online social media and offline in-person networks, we introduce a second network $G_2 = (\mathcal{V}, \mathcal{E}_2, W_2)$, sharing the same vertex set $\mathcal{V}$ with $G_1$, but with the edge set $\mathcal{E}_2 \subseteq \mathcal{V} \times \mathcal{V}$ and weight matrix $W_2$ being distinct from those of $G_1$. Each edge $(i, j) \in \mathcal{E}_2$ has a positive weight $w_{ij}^{(2)} > 0$, with $w_{ij}^{(2)} = 0$ for edges not in $\mathcal{E}_2$, and no self-loops. 

Similar to the row-stochastic matrix $A$ defined earlier, we define a row-stochastic transition matrix $B$ for $G_2$ with entries
\begin{equation*}
B_{ij} =
\begin{cases}
\frac{w_{ij}^{(2)}}{\sum_{k \in \partial_i^{(2)}} w_{ik}^{(2)}} & \text{if } (i, j) \in \mathcal{E}_2, \\
0 & \text{otherwise}.
\end{cases} 
\end{equation*}
We also assume $B$ is primitive, ensuring convergence to consensus in $G_2$.

We now define the \emph{merged network} as $G_m = (\mathcal{V}, \mathcal{E}_m, W_m)$, with the weight matrix $W_m = \alpha W_1 + (1 - \alpha) W_2$ and \(\mathcal{E}_m= \mathcal{E}_1 \cup \mathcal{E}_2 \), where $\alpha \in (0, 1)$ balances the influence of the two layers. 
The utility function for agent $i$ over $G_m$ integrates influences from both layers:
\begin{equation*}
J_i(x_i, \mathbf{x}_{\partial_i^{(m)}}) = \frac{1}{2} \sum_{j \in \partial_i^{(m)}} w_{ij}^{(m)} (x_i - x_j)^2, 
\end{equation*}
where $\partial_i^{(m)}$ 
is the out-neighborhood of agent $i$ in $G_m$, and with edge weights $w_{ij}^{(m)} = \alpha w_{ij}^{(1)} + (1 - \alpha) w_{ij}^{(2)}$. 
Similar to the single-layer network, we first note that the best-response strategy is
$x_i(t+1)=\frac{\sum_{j \in \partial_i^{(m)}}\big(\alpha {w_{ij}^{(1)}} + (1-\alpha) w_{ij}^{(2)}\big)x_j(t)}{\sum_{k \in \partial_i^{(m)}} \alpha w_{ik}^{(1)} + (1-\alpha ) w_{ik}^{(2)}}$. 
Accordingly, the opinion dynamics can be interpreted as a process over the merged graph $G_m$, defined by the combined edge set $\mathcal{E}_m$ and weight matrix $W_m$, and opinion dynamics
\begin{equation}
\mathbf{x}(t+1) =C \mathbf{x}(t)~,
\label{eq:dynamics-merged}
\end{equation} 
where $C$ is the merged transition matrix given by
\[C_{ij} = \begin{cases}
\frac{w_{ij}^{(m)}}{\sum_{k \in \partial_i^{(m)}} w_{ik}^{(m)}} & \text{if } (i, j) \in \mathcal{E}_m, \\
0 & \text{otherwise}.
\end{cases} \]

\subsection{Consensus}
\label{subsec:consensus_opinion}

We begin with the existence and uniqueness of consensus. 

\begin{proposition}
\label{prop:merge_primitive}
Let $A$ and $B$ be the transition matrices associated with layers $G_1$ and $G_2$, respectively. Suppose that at least one of $A$ or $B$ is primitive, and let $\alpha \in (0,1)$. Then the merged transition matrix $C$ is primitive.
\end{proposition}

The proof follows from algebraic manipulations of the definitions of the matrices and the assumed primitivity of one of them. 
%
Intuitively, this finding establishes that stable consensus can emerge even when one interaction layer has poor mixing properties, provided the other layer compensates. 

We next compare the consensus opinion of the merged network $G_m$ with those of the individual layers $G_1$ and $G_2$. 

\begin{proposition}
\label{prop:consensus_bound}
The merged consensus $x^m(\infty)$ satisfies
\[
\min\{x^1(\infty),\,x^2(\infty)\}
\;\leq\;
x^m(\infty)
\;\leq\;
\max\{x^1(\infty),\,x^2(\infty)\},
\]
where $x^1(\infty)$ (resp. $x^2(\infty)$) is $G_1$'s (resp. $G_2$'s) consensus. 
\end{proposition}

The proof follows primarily from definitions. This finding states that the consensus of the merged system always lies between the consensuses of the individual layers. In fact, the bounds are tight: if the two layers are identical, then all three consensus values coincide exactly. 

\vspace{-0.1in}
\subsection{Convergence Rate}
\label{sec:convergence_rate} 

We next characterize the convergence rate to consensus. 

\begin{proposition}
\label{prop:slem_bound}
The SLEM of merged matrix $C$ satisfies
$\rho_2(C) \geq \frac{1}{N-1}$. 
Moreover, if $d^{(1)}_{i} = d^{(2)}_{i}$ for all $1 \leq i \leq N$, then
\begin{equation*}
\rho_2(C) \leq \max\{\rho_2(A), \rho_2(B)\}.
\end{equation*}
\end{proposition}

\noindent\textit{Proof sketch.} 
For the lowerbound, we apply Vieta's formulas to the characteristic polynomial, followed by noting the Perron eigenvalue of $\lambda_1(C)=-1$ and invoking the pigeonhole principle. For the upperbound, we first establish a special relation between $A, B, C$ under the degree-matching condition, and then apply the Courant–Fischer Theorem to bound the SLEM.

\noindent\textit{Intuitive interpretation.} First, the lower bound $\tfrac{1}{N-1}$ states that convergence cannot be made arbitrarily fast, regardless of connectivity. While one might expect a sharper bound such as $\min\{\rho_2(A), \rho_2(B)\}$, this fails in general since the maximizer of the Rayleigh quotient for $C$ may drive $r_C(v)$ arbitrarily close to zero. On the other hand, obtaining a favorable upper bound (i.e, keeping the SLEM of the merged matrix $C$ small), requires matching degree sequences; notably, 
conflicting degree structures between layers slows down mixing.

To complete our analysis of convergence rate in this model, we provide Example~\ref{ex:necessity} below to demonstrate that the degree-matching condition $d^{(1)}_{i} = d^{(2)}_{i}, \forall i$ in Proposition~\ref{prop:slem_bound} is necessary for the merged matrix $C$'s SLEM upper bound, and Example~\ref{ex:convergence_rate} below to show that the lower bound is tight. 

\begin{example}
\label{ex:necessity}
Let $\alpha = \frac{1}{2}$, and consider the matrices
\begin{scriptsize}
\begin{align*}
W_1 = \begin{pmatrix}
0  & 50 &  1 &  1 &  2 & 40\\
50 &  0 &  3 &  1 & 50 & 50\\
1  &  3 &  0 & 40 & 40 &  2\\
1  &  1 & 40 &  0 & 40 &  3\\
2  & 50 & 40 & 40 &  0 &  1\\
40 & 50 &  2 &  3 &  1 &  0
\end{pmatrix}, 
W_2 = \begin{pmatrix}
0  &  1 &  3 &  1 &  1 &  1\\
1  &  0 &  1 &  2 &  1 &  3\\
3  &  1 &  0 & 50 & 40 &  3\\
1  &  2 & 50 &  0 & 50 &  2\\
1  &  1 & 40 & 50 &  0 &  1\\
1  &  3 &  3 &  2 &  1 &  0
\end{pmatrix}
\end{align*}
\end{scriptsize}
Here, the degrees $d_i(j)$ differ between the two layers. Also, 
\begin{equation*}
0.6928 =\rho_2(C)> \max\{ \rho_2(A), \rho_2(B)\}=\max\{ 0.6839, 0.5338\}.
\end{equation*}
which violates the upper bound and confirms the necessity of the matching degree condition. 
\end{example}

\begin{example}
\label{ex:convergence_rate}
Let $\alpha = \frac{1}{2}$, and consider the matrices
{\begin{scriptsize}
\begin{align*}
    W_1 = \begin{pmatrix}
0 & \frac{1}{2} & 0 & 0 & \frac{1}{2} \\
\frac{1}{2} & 0 & \frac{1}{2} & 0 & 0 \\
0 & \frac{1}{2} & 0 & \frac{1}{2} & 0 \\
0 & 0 & \frac{1}{2} & 0 & \frac{1}{2} \\
\frac{1}{2} & 0 & 0 & \frac{1}{2} & 0
\end{pmatrix}, ~~~
 W_2= \begin{pmatrix}
0 & 0 & \frac{1}{2} & \frac{1}{2} & 0 \\
0 & 0 & 0 & \frac{1}{2} & \frac{1}{2} \\
\frac{1}{2} & 0 & 0 & 0 & \frac{1}{2} \\
\frac{1}{2} & \frac{1}{2} & 0 & 0 & 0 \\
0 & \frac{1}{2} & \frac{1}{2} & 0 & 0
\end{pmatrix}.
\end{align*}
\end{scriptsize}}
Both matrices have eigenvalues $\cos\left(\frac{2k\pi}{5}\right)$ for $k = 0, 1, 2, 3, 4$. Then, the SLEM of $C$ is $\rho_2(C) = \frac{1}{4}$, achieving the lower bound $\frac{1}{N-1}$ with equality. This demonstrates that merging two slow-mixing sparse layers can yield a fast-mixing dense layer. 
\end{example}

\subsection{Stability of Opinion Dynamics}
\label{sec:stability_merged}

Lastly, we establish two complementary stability results for the merged model. 
First, we show that when $\alpha$ approaches 1, so that layer $G_2$ is effectively ignored, the merged consensus converges \emph{linearly} to the single-layer consensus of layer $G_1$. 
\begin{proposition}
\label{prop:alpha_stability}
Let $x^1(\infty)$ denote the consensus of layer $A$ in isolation. As $\alpha\to1^-$, 
\[\bigl|x^m(\infty)-x^1(\infty)\bigr| = O\bigl(1-\alpha\bigr).\]
\end{proposition}

Second, building on Corollary~\ref{cor:stability_bound}, we find that if $B$ is a small perturbation of $A$, the change in consensus is of the same order as the matrix perturbation. 
\begin{proposition}
\label{prop:perturbation_merged}
If $\|A - B\|_{\max}$ is sufficiently small, then
\[
\bigl|x^m(\infty) - x^1(\infty)\bigr|
= O\bigl(\|A - B\|_{\max}\bigr).
\]
\end{proposition}
\section{Switching Layers Opinion Dynamics}
\label{switching}

\subsection{Switching Layers Coordination Game Model}\label{subsec:switched_multilayer}

In an alternative view of dynamic social systems, agents can be seen as \emph{alternating} between interaction modalities (e.g., more commonly interacting over social media platforms and only intermittently exchanging opinions in person). We model this behavior with a periodic switching mechanism between $G_1$ and $G_2$, where $G_1$ and $G_2$ are defined as earlier.

Specifically, we assume agents' opinions evolve in $G_1$ under matrix $A$ for $k$ steps. Then, agents switch to $G_2$ and their opinions evolve under matrix $B$ for one step, before again switching back to $G_1$, with this alternating process repeated. The cost function for agent $i$ at time $t$ is defined as being dependent only on the active layer:
\begin{align*}
J_i(x_i, \mathbf{x}_{\partial_i^{(1)}}, \mathbf{x}_{\partial_i^{(2)}}, t) &= \notag\\
&\hspace{-1.1in}\begin{cases}
\frac{1}{2} \sum_{j \in \partial_i^{(1)}} w_{ij}^{(1)} (x_i - x_j)^2 & \text{if } (t+1) \bmod (k+1) \neq 0, \\
\frac{1}{2} \sum_{j \in \partial_i^{(2)}} w_{ij}^{(2)} (x_i - x_j)^2 & \text{otherwise},
\end{cases} 
\end{align*}
where $k+1 \geq 2$ is the \emph{switching period}. Larger $k$ reflect scenarios where platform $G_1$ dominates agents' time. 

Similar to before, we can find the best-response strategies, following which the opinion dynamics can be expressed as
\begin{equation}
\mathbf{x}(t) =
\begin{cases}
A \mathbf{x}(t-1) & \text{if } t \bmod (k+1) \neq 0, \\
B \mathbf{x}(t-1) & \text{otherwise},
\end{cases} \label{eq:switching_dynamics}
\end{equation}
with the opinion vector at time $t$:
\begin{equation}
\mathbf{x}(t) = \big( \prod_{s=1}^t M(s) \big) \mathbf{x}(0), \label{eq:switching_iterated_dynamics}
\end{equation}
where 
$M(s) =
\begin{cases}
A & \text{if } s \bmod (k+1) \neq 0, \\
B & \text{otherwise}.
\end{cases}$

\subsection{Consensus}
\label{subsec:switching_consensus}

We first show that in this switching model, even if $A$ and $B$ are primitive (which ensures convergence to a consensus on each layer in isolation), the product $\prod_{s=1}^t M(s)$ in \eqref{eq:switching_iterated_dynamics} may not converge because primitivity is not preserved under multiplication, thereby impeding consensus on the multiplex network. The following example illustrates this.

\begin{example}
\label{ex:non_convergence}
Consider a switching period $k+1=2$ with adjacency matrices for $G_1$ and $G_2$:

{\begin{scriptsize}
\begin{align*}
W_1 &= \begin{pmatrix}
0 & 1 & 1 & 1 & 0 \\
1 & 0 & 0 & 1 & 0 \\
1 & 0 & 0 & 0 & 0 \\
1 & 1 & 0 & 0 & 1 \\
0 & 0 & 0 & 1 & 0
\end{pmatrix}, \quad
W_2 = \begin{pmatrix}
0 & 0 & 0 & 0 & 1 \\
0 & 0 & 1 & 0 & 1 \\
0 & 1 & 0 & 1 & 1 \\
0 & 0 & 1 & 0 & 0 \\
1 & 1 & 1 & 0 & 0
\end{pmatrix}.
\end{align*}
\end{scriptsize}}
A direct calculation shows that the product $\prod_{s=1}^t M(s)$ does not converge, hence consensus does not exist.
\end{example}
 
Despite this negative result, there are still sufficient conditions under which consensus can emerge in switching layer coordination games, as shown by the following proposition. 

\begin{proposition}
\label{prop:switching_consensus}
Assume the matrix $BA^k$ is primitive with a unique stationary distribution $\boldsymbol{\pi}$. 
Then, the switching dynamics \eqref{eq:switching_dynamics} converges to consensus with stationary distribution $\boldsymbol{\pi}$. 
\end{proposition} 
\noindent\textit{Proof sketch.} The proof leverages mixing-time results (\cite{bremaud2013markov, aldous-fill-2014, levin2017markov}) to bound $\|A^r(BA^k)^t - \mathbf{1}\boldsymbol{\pi}^T\|_{\max}$, the deviation of the opinion at any time $t$ from $\boldsymbol{\pi}$. It then invokes the Perron-Frobenius theorem to show that this bound converges to zero exponentially, establishing convergence to consensus $\boldsymbol{\pi}$. 

\noindent\textit{Intuitive interpretation.} $BA^{k}$ represents one communication round: $k$ steps on $G_1$ followed by one step on $G_2$. Primitivity of this combined operator is sufficient for consensus, which is a weaker condition than requiring primitivity of $A$ and $B$ individually. This is because alternating between layers can supply the connectivity that neither layer possesses alone. 

We next take a closer look at the nature of the resulting consensus. The stationary distribution \(\boldsymbol{\pi}\) of \(BA^{k}\) can be unexpectedly complex, as illustrated in Example~\ref{ex:stationary_distribution}: its components do not necessarily lie between those of \(\boldsymbol{\pi}_A\) and \(\boldsymbol{\pi}_B\), the consensus on each isolated layer.  

\begin{example}
\label{ex:stationary_distribution}
Consider adjacency matrices

{\begin{scriptsize}
\[W_1 = \begin{pmatrix}
0 & \tfrac12 & \tfrac12 \\
\tfrac12 & 0 & \tfrac12 \\
\tfrac12 & \tfrac12 & 0
\end{pmatrix}, 
\qquad
W_2 = \begin{pmatrix}
0 & 2 & 1 \\
2 & 0 & 1 \\
1 & 1 & 0
\end{pmatrix}.\]
\end{scriptsize}}

Then, for $k=1$, we have
\[
\boldsymbol{\pi}_A = \Bigl( \tfrac13, \tfrac13, \tfrac13 \Bigr), 
\
\boldsymbol{\pi}_B = \Bigl( \tfrac38, \tfrac38, \tfrac14 \Bigr), \boldsymbol{\pi}_{BA} = \Bigl( \tfrac{3}{10}, \tfrac{3}{10}, \tfrac{2}{5} \Bigr).
\]
Observe that
\[
\pi_{BA}(i) \;\notin\; 
\bigl[\min\{\pi_A(i),\pi_B(i)\},\,\max\{\pi_A(i),\pi_B(i)\}\bigr],
\]
so \(\boldsymbol{\pi}_{BA}\) does not interpolate between \(\boldsymbol{\pi}_A\) and \(\boldsymbol{\pi}_B\).
\end{example}

\subsection{Convergence Rate}
\label{sec:switching_convergence_rate}

As established in Proposition~\ref{prop:switching_consensus}, the matrix $BA^k$ governs the existence of consensus in the switching layer. Consequently, the convergence rate of the switching dynamics is determined by the SLEM of $BA^k$. Accordingly, the error
$\mathbf{e}(t)$
decays by a factor of $\rho_2(BA^k)$ per each length $k+1$ cycle. Thus, the effective per-step convergence rate is approximately $\sqrt[k+1]{\rho_2(BA^k)}$. To rigorously quantify the convergence rate, we invoke the following result.

\begin{proposition}[\cite{bajovic2013consensus}, Proposition 8]
\label{prop:product_slem}
Let $P_1, \dots, P_n$ be stochastic, reversible matrices with stationary distributions $\mu_1, \dots, \mu_n$. Define the product matrix
\(
P_{\mathrm{prod}}^n = P_n P_{n-1} \cdots P_1.
\)
Then, the SLEM satisfies
\begin{align*}
\left| \rho_2 \left( P_{\mathrm{prod}}^n \right) \right|
&\leq\\
\hspace{-0.32in}&
\left( \prod_{i=1}^n \left| \rho_2(P_i) \right| \right)
\left( \prod_{i=2}^n \max_j \frac{\mu_{i-1}(j)}{\mu_i(j)} \right)
\max_j \frac{\mu_n(j)}{\mu_1(j)}.
\end{align*}
\end{proposition}

Applying this to our model leads to a convergence bound. 

\begin{proposition}
\label{prop:switching_convergence}
Under the alternating application of matrix $A$ for $k$ steps followed by matrix $B$, the error satisfies
\[
\|\mathbf{e}(t)\|_{\max} \lesssim \rho_*^{\left\lfloor \frac{t}{k+1} \right\rfloor} \|\mathbf{e}(0)\|_\pi,
\]
where $\rho_* = \rho_2(B) \cdot \rho_2(A)^k \cdot \max_i \frac{d^{(1)}_{i}}{d^{(2)}_{i}} \cdot \max_i \frac{d^{(2)}_{i}}{d^{(1)}_{i}}$. 
\end{proposition}

\noindent\textit{Intuitive interpretation.} This result quantifies how quickly consensus is reached when agents switch between two platforms. Each full cycle of $k$ steps on $G_1$ followed by one step on $G_2$ shrinks disagreement by about $\rho_*$. The terms $\rho_2(A)^k$ and $\rho_2(B)$ capture how strongly each layer mixes opinions during its turn, while the degree-ratio factors measure structural alignment: 
if their influence patterns differ, the ratios inflate and slow down mixing. 

\subsection{Stability of the Opinion Dynamics}
\label{sec:stability_switched}

Lastly, we analyze two forms of stability for the switching layer opinion dynamics. The first examines the limit as $k \to \infty$, where the dynamics are dominated by layer $G_1$. 

\begin{proposition}\label{prop:pi_convergence}
Let $\boldsymbol{\pi}^k$ be the stationary distribution of $BA^k$ and let $x^{s,k}(\infty)$ denote the consensus value in the switching layer with period $k+1$.  Then, as $k \to \infty$,
\begin{equation*}
\bigl|x^{s,k}(\infty)-x^1(\infty)\bigr| = O\!\bigl(\rho_2(A)^k\bigr).
\end{equation*}
\end{proposition}

The second evaluates robustness when there are minor variations in network structure between layers. 

\begin{proposition}
\label{prop:perturbation_stability}
If $\|A - B\|_{\max}$ is sufficiently small, then 
\begin{equation*}
\bigl|x^s(\infty) - x^1(\infty)\bigr| = O\bigl(\|A - B\|_{\max}\bigr).
\end{equation*}
\end{proposition}
\section{Numerical Experiments}
\label{sec:numerical_results}

We validate our theoretical findings through experiments on a real-world multiplex high-school contact data from \cite{mastrandrea2015contact}. 
Layer $A$ encodes Facebook friendships ($w_{ij}=1$ if students $i$ and $j$ are friends, 0 otherwise). Layer $B$ captures face-to-face contact durations via $w_{ij} \in \{1,2,3,4\}$, corresponding to contact intervals of $\leq 5$ min, 5–15 min, 15 min–1h, and $>1$h, respectively. Initial opinions are drawn uniformly from $[0,1]$. We select 72 students so that layer $B$ alone does not reach consensus. 
Figures~\ref{fig:merged_layer_consensus} and~\ref{fig:switching_layer_consensus} confirm our theoretical finding that cross-layer interactions induce a stable consensus even when individual layers lack this property (as seen in $\alpha=0$ or $k=0$).\looseness-1

\begin{figure*}[t]
\centering
\subfloat[\(\alpha=0\)]{%
  \includegraphics[width=.3\textwidth]{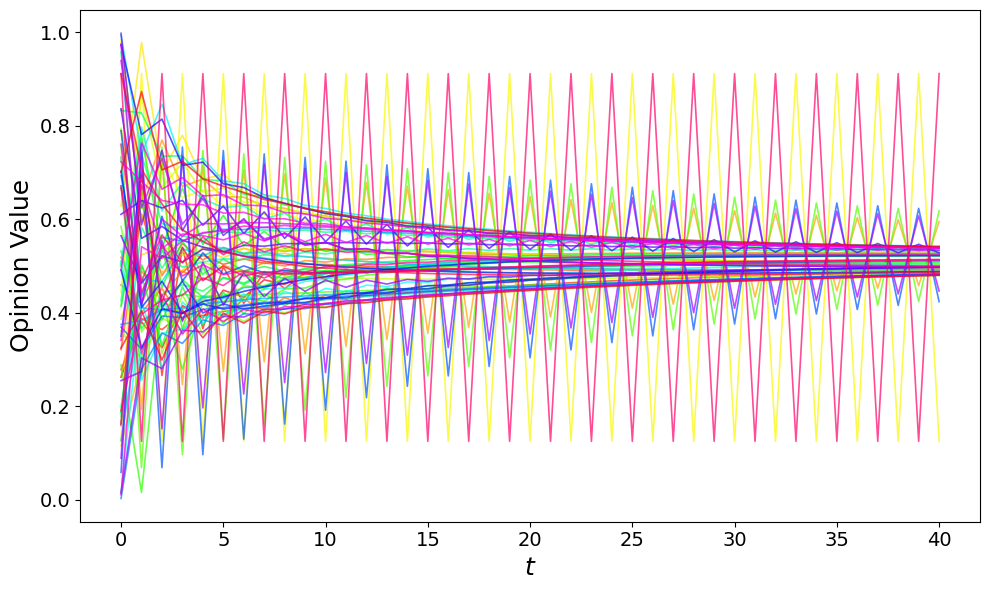}%
  \label{fig:real_merged_a}
}\hfill
\subfloat[\(\alpha=0.5\)]{%
  \includegraphics[width=.3\textwidth]{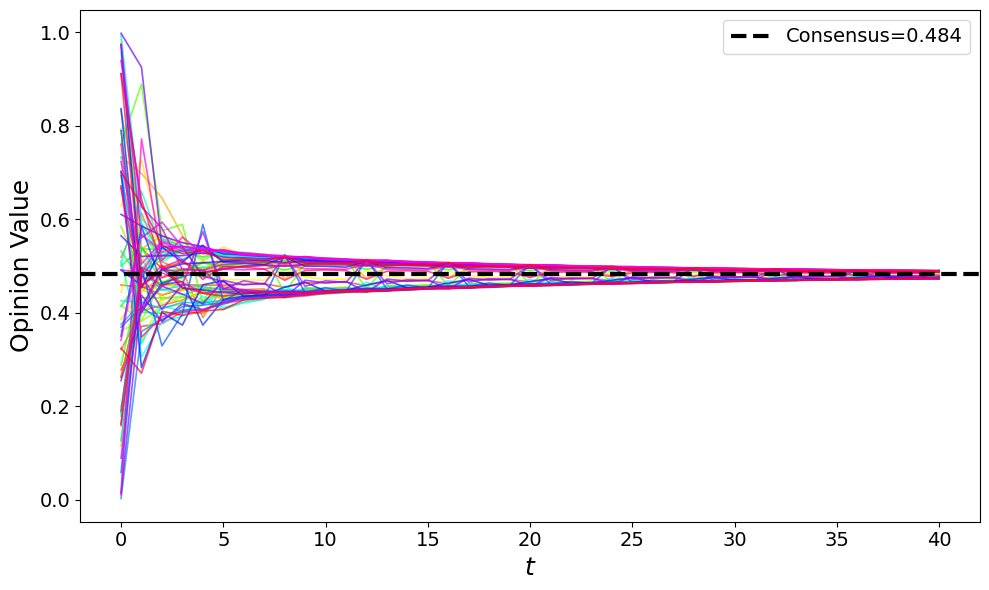}%
  \label{fig:real_merged_b}
}\hfill
\subfloat[\(\alpha=1\)]{%
  \includegraphics[width=.3\textwidth]{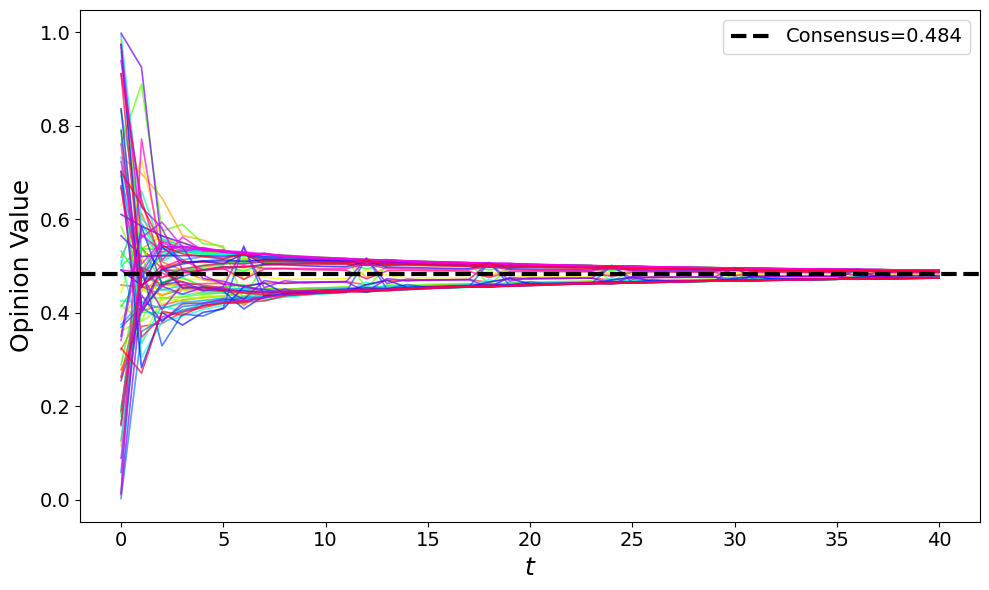}%
  \label{fig:real_merged_c}
}
\caption{Merged layers opinion dynamics on the high-school contact network \cite{mastrandrea2015contact} for different weighting factors $\alpha$.}
\label{fig:merged_layer_consensus}
\end{figure*}
\vspace{-0.1in}
\begin{figure*}[t]
\centering
\subfloat[\(k=0\)]{%
  \includegraphics[width=.3\textwidth]{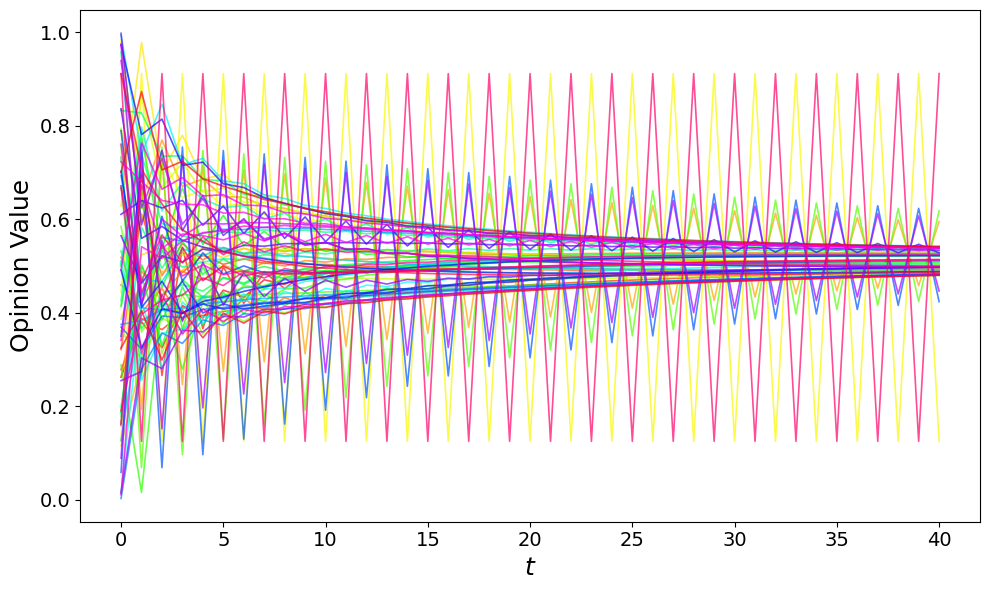}%
  \label{fig:real_switching_a}
}\hfill
\subfloat[\(k=3\)]{%
  \includegraphics[width=.3\textwidth]{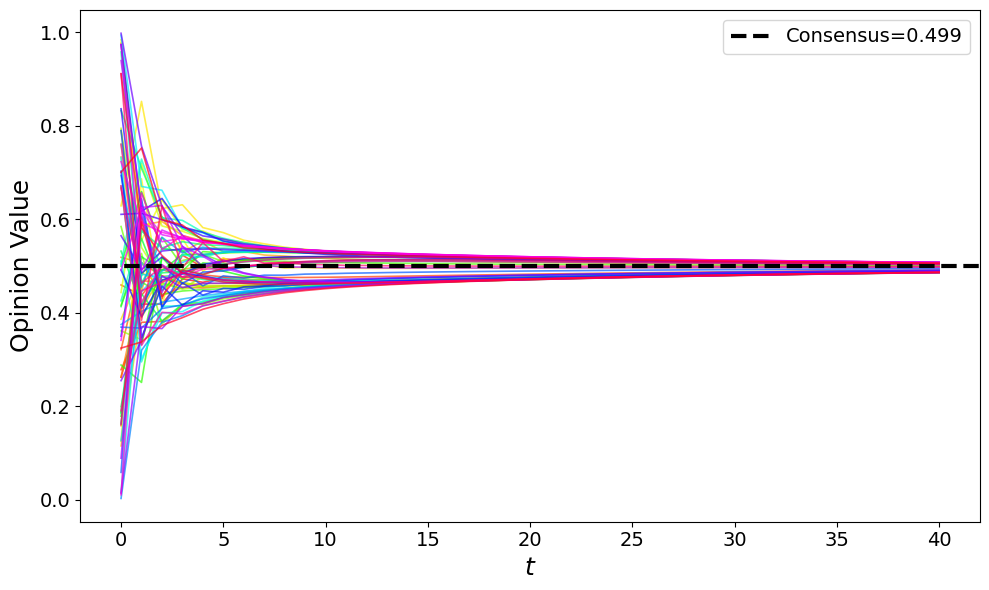}%
  \label{fig:real_switching_b}
}\hfill
\subfloat[\(k=5\)]{%
  \includegraphics[width=.3\textwidth]{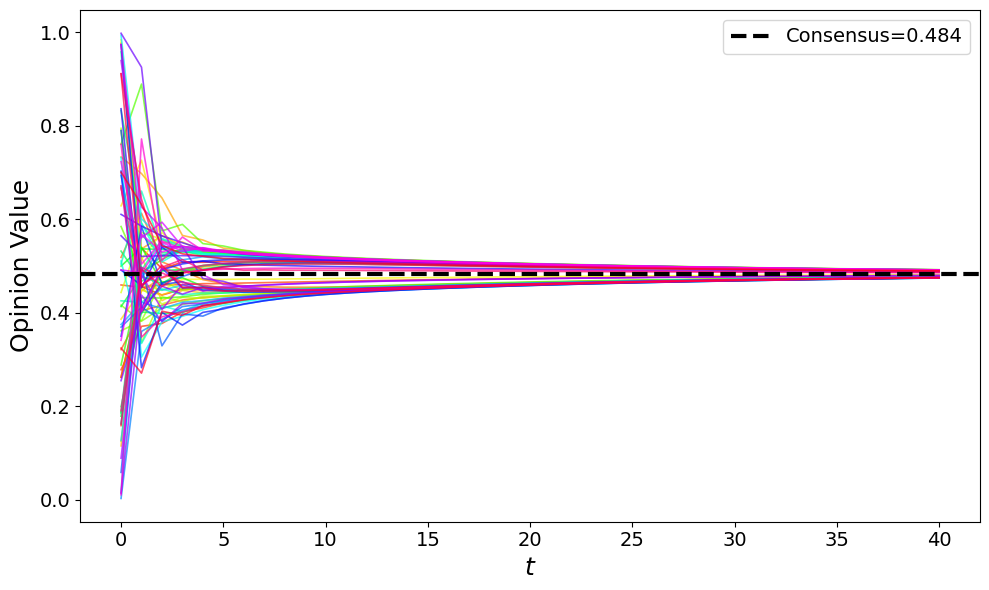}%
  \label{fig:real_switching_c}
}
\caption{Switching layers opinion dynamics on the high-school contact network \cite{mastrandrea2015contact} with different switching period $k+1$.}
\vspace{-0.1in}
\label{fig:switching_layer_consensus}
\end{figure*}
\section{Conclusion}\label{sec:conclusion}

We proposed two complementary models for multilayer opinion dynamics: merged layers (to capture simultaneous exposure) and switching layers (to capture shifting attention), and analyzed consensus, convergence, and stability of opinion dynamics for each of these models.
Our key finding is that multilayer interactions can create (resp. undermine) consensus even when individual layers fail to converge (resp. support convergence in isolation). 
Notably, we show that the alignments in the weighted degrees of the nodes, a notion of similarity between the layers, is a main determinant of if merging or switching can speed up convergence to consensus compared to when the layers operate in isolation. These results can provide design and intervention guidelines to help influence opinion dynamics on interacting social platforms.

Future research directions include extending our framework to adaptive layer interactions where weights or switching patterns vary over time or across social contexts, and incorporating heterogeneous agent behaviors, including stubborn agents, agents with varying susceptibilities to different layer influences, or those adopting different switching patterns.

\bibliographystyle{IEEEtran}
\bibliography{ref} 

\clearpage

\section{Appendix}\label{sec:appendix}
\subsection{Proof of Proposition \ref{prop:merge_primitive}}

\begin{proof}
Without loss of generality, assume that $A$ is primitive. Then there exists a positive integer $n$ such that $A^n > 0$, 
where $M>0$ denotes that all entries of the matrix $M$ are strictly positive. By the definition of $A$ and $B$ as transition matrices, it follows that $A^m \ge 0$ and $B^m \ge 0$ for all $m \ge 1$, and the same holds for all finite products of these matrices.

Recall that the entries of the merged transition matrix $C$ are given by 
\[C_{ij} = \dfrac{\alpha w_{ij}^{(1)} + (1-\alpha) w_{ij}^{(2)}}
{\sum_{k \in \partial_i^{(m)}} \alpha w_{ik}^{(1)} + (1-\alpha) w_{ik}^{(2)}}~,\]
when \(C_{ij} \neq 0.\)
Using the fact that all weights are nonnegative, we obtain the lower bound
\begin{align*}
    C_{ij}
&\ge
\dfrac{\alpha w_{ij}^{(1)}}{2(|E_1|+|E_2|)}
=
\dfrac{\alpha \sum_{k \in \partial_i^{(1)}} w_{ik}^{(1)}}{2(|E_1|+|E_2|)} A_{ij} \\
&\ge
\dfrac{\alpha \min_{1\le i \le N} d_i^{(1)}}{2(|E_1|+|E_2|)} A_{ij}~,
\end{align*}
where the inequalities follow primarily from the definitions $d^{(k)}_{i} = \sum_{j=1}^N w_{ij}^{(k)}$ and $|E_k| = \frac{1}{2} \sum_{i=1}^N d^{(k)}_{i}$. 

When \(C_{ij}=0\), we also have \(A_{ij}=0\). Consequently,
\[ C - \frac{\alpha \min_{1\le i \le N} d_i^{(1)}}{2(|E_1|+|E_2|)} A \ge 0 \]
entrywise. We can therefore write
\begin{align*}
C^n
&=
\left(
\frac{\alpha \min_{1\le i \le N} d_i^{(1)}}{2(|E_1|+|E_2|)} A
+
\left(
C - \frac{\alpha \min_{1\le i \le N} d_i^{(1)}}{2(|E_1|+|E_2|)} A
\right)
\right)^n \\
&\ge
\left(
\frac{\alpha \min_{1\le i \le N} d_i^{(1)}}{2(|E_1|+|E_2|)}
\right)^n
A^n
>
0,
\end{align*}
where the inequality follows from the nonnegativity of all terms and the strict positivity of $A^n$. Hence, $C^n>0$, which implies that $C$ is primitive. This establishes the existence and uniqueness of consensus in the merged network.
\end{proof}

\subsection{Proof of Proposition \ref{prop:consensus_bound}}

\begin{proof}
For each layer $G_\ell$ ($\ell=1,2$), the consensus is
\[
x^\ell(\infty)
=\frac{1}{2|E_\ell|}\sum_{i=1}^N d_{ i}^{(\ell)}x_i(0).
\]
In the merged network $G_m$, we obtain
\begin{align*}
x^m(\infty)
&=\frac{\sum_{i=1}^N \bigr(\alpha d_{ i}^{(1)} +(1-\alpha)d_{ i}^{(2)}\bigl) x_i(0)}{2\bigl(\alpha|E_1|+(1-\alpha)|E_2|\bigr)}
\nonumber \\
&=\frac{\alpha|E_1|\,x^1(\infty)+(1-\alpha)|E_2|\,x^2(\infty)}
{\alpha|E_1|+(1-\alpha)|E_2|}.
\end{align*}

Since this is a convex combination of $x^1(\infty)$ and $x^2(\infty)$, we immediately get
\[
\min\{x^1(\infty),x^2(\infty)\}
\;\leq\;
x^m(\infty)
\;\leq\;
\max\{x^1(\infty),x^2(\infty)\}.
\]
\end{proof}

\subsection{Proof of Proposition \ref{prop:slem_bound}}

\begin{proof}
For the lower bound, since $C$ has zero diagonal entries, applying Vieta's formulas to the characteristic polynomial yields
\begin{equation*}
\operatorname{tr}(C) = \sum_{i=1}^N \lambda_i(C) = 0.
\end{equation*}
Given that $\lambda_1(C) = 1$ is the Perron eigenvalue, we have
\begin{equation*}
\sum_{i=2}^N \lambda_i(C) = -1.
\end{equation*}
By the pigeonhole principle, there exists some $j \geq 2$ such that
\begin{equation*}
\lambda_j(C) \leq -\frac{1}{N-1}.
\end{equation*}
Therefore,
\begin{equation*}
\rho_2(C) = \max_{i \geq 2} |\lambda_i(C)| \geq \frac{1}{N-1}.
\end{equation*}

We now turn to the upper bound.
Assume that $d^{(1)}_{i} = d^{(2)}_{i}$ for all $1 \le i \le N$, that is,
\begin{equation*}
\sum_{k \in \partial_i^{(1)}} w_{ik}^{(1)} 
= 
\sum_{k \in \partial_i^{(2)}} w_{ik}^{(2)}
\qquad \text{for all } i.
\end{equation*}

For any $i,j$, we have
\begin{align*}
\alpha A_{ij} &+ (1-\alpha) B_{ij} \\
&=
\alpha\left(\frac{w_{ij}^{(1)}}{\sum_{k \in \partial_i^{(1)}} w_{ik}^{(1)}}\right)
 + (1-\alpha)\left(\frac{w_{ij}^{(2)}}{\sum_{k \in \partial_i^{(2)}} w_{ik}^{(2)}}\right) \\
&= \frac{\alpha w_{ij}^{(1)} + (1-\alpha) w_{ij}^{(2)}}{\sum_{k \in \partial_i^{(1)}} w_{ik}^{(1)}} \\
&= \frac{\alpha w_{ij}^{(1)} + (1-\alpha) w_{ij}^{(2)}}{\sum_{k \in \partial_i^{(m)}} 
\bigl(\alpha w_{ik}^{(1)} + (1-\alpha) w_{ik}^{(2)}\bigr)} = C_{ij}~.
\end{align*}
Hence, under the degree-matching assumption $C = \alpha A + (1-\alpha) B$. 

Now, let $D$ be the diagonal matrix with entries $d^{(1)}_{i}$.  
We may therefore write
\begin{equation*}
A = D^{-1}W_A, \quad B = D^{-1}W_B, \quad C = D^{-1}W_C,
\end{equation*}
where $W_C = \alpha W_A + (1-\alpha)W_B$ with $0 \leq \alpha \leq 1$.

Define the symmetrized matrices
\begin{equation*}
S_X = D^{-\frac{1}{2}} W_X D^{-\frac{1}{2}}, \quad X \in \{A, B, C\}.
\end{equation*}
Each matrix $S_X$ is real symmetric and similar to $X$, hence they share the same eigenvalues. Since $A$, $B$, and $C$ are stochastic matrices, each $S_X$ has largest eigenvalue $\lambda_1 = 1$ with corresponding eigenvector $D^{1/2}\mathbf{1}_N$. Moreover,
\[
S_C = \alpha S_A + (1-\alpha) S_B.
\]

For any unit vector $v \perp D^{\frac{1}{2}}\mathbf{1}_N$, define the Rayleigh quotient (~\cite{horn2012matrix}, Theorem 4.2.2)
\begin{equation*}
r_X(v) = v^T S_X v.
\end{equation*}
By the Courant--Fischer theorem (~\cite{horn2012matrix}, Theorem 4.2.6),
\begin{equation*}
\lambda_2(C) = \max_{\substack{\|v\| = 1 \\ v \perp D^{\frac{1}{2}}\mathbf{1}_N}} r_C(v).
\end{equation*}
Let $v^*$ be the vector achieving this maximum. Then
\begin{align}
\label{1234}
\lambda_2(C) = r_C(v^*) &= (v^*)^T S_C v^* \nonumber \\
&= \alpha (v^*)^T S_A v^* + (1-\alpha)(v^*)^T S_B v^* \nonumber \\
&= \alpha r_A(v^*) + (1-\alpha) r_B(v^*) \nonumber \\
&\leq \max\{r_A(v^*), r_B(v^*)\} \nonumber \\
&\leq \max\{\lambda_2(A), \lambda_2(B)\}.
\end{align}
Similarly, we can show that
\begin{equation}
\label{12345}
\min\{\lambda_N(A), \lambda_N(B)\} \leq \lambda_N(C).
\end{equation}
Combining \eqref{1234} and \eqref{12345} establishes the upper bound on $\rho_2(C)$ in the proposition statement. 
\end{proof}

\subsection{Proof of Proposition \ref{prop:alpha_stability}}

\begin{proof}
As shown in the proof of Proposition~\ref{prop:consensus_bound},
\[
x^m(\infty)-x^1(\infty)
=\frac{(1-\alpha)\,|E_2|\,\bigl(x^2(\infty)-x^1(\infty)\bigr)}
{\alpha\,|E_1|+(1-\alpha)\,|E_2|}\,.
\]
Since 
$\alpha\,|E_1|+(1-\alpha)\,|E_2|\geq\min\{|E_1|,|E_2|\}$, we obtain
\begin{align*}
    \bigl|x^m(\infty)-x^1(\infty)\bigr|
&\leq\frac{(1-\alpha)\,|E_2|}{\min\{|E_1|,|E_2|\}}
\;\bigl|x^2(\infty)-x^1(\infty)\bigr|
\\&= C\,(1-\alpha),
\end{align*}
where 
\[
C=\frac{|E_2|}{\min\{|E_1|,|E_2|\}}\,
\bigl|x^2(\infty)-x^1(\infty)\bigr|.
\]
\end{proof}

\subsection{Proof of Proposition \ref{prop:switching_convergence}}

\begin{proof}
Since graphs $G_1$ and $G_2$ are undirected, the matrices $A$ and $B$ are reversible. Given that $A$ is primitive, so is $A^k$, and both share the same stationary distribution $\boldsymbol{\pi}$, satisfying $\boldsymbol{\pi} A^k = \boldsymbol{\pi}$. Let $\mathbf{v}_2$ be the eigenvector of $A$ corresponding to $\rho_2(A)$. Then,
\begin{equation*}
A^k \mathbf{v}_2 = \rho_2(A)^k \mathbf{v}_2,
\end{equation*}
implying $\rho_2(A^k) = \rho_2(A)^k$. By Proposition~\ref{prop:product_slem}, for reversible matrices $B$ and $A^k$, and using the relationship between stationary distributions and degree sequences from \eqref{eq:stationary_dist}, the SLEM of the product satisfies
\begin{equation*}
\rho_2(BA^k) \leq \rho_* = \rho_2(B) \cdot \rho_2(A)^k \cdot \max_i \frac{d^{(1)}_{i}}{d^{(2)}_{i}} \cdot \max_i \frac{d^{(2)}_{i}}{d^{(1)}_{i}}.
\end{equation*}
For $t = n(k+1)$, the error after $n$ cycles is
\begin{align*}
\|\mathbf{e}(n(k+1))\|_{\max} &= \| (BA^k)^n \mathbf{e}(0) \|_{\max}  \notag \\&\lesssim \rho_2(BA^k)^n \|\mathbf{e}(0)\|_\pi \notag \\
&\lesssim  \rho_*^n \|\mathbf{e}(0)\|_\pi.
\end{align*}
For general $t = n(k+1) + r$, where $0 \leq r \leq k$, similar to the proof in Proposition \ref{prop:switching_consensus}, the additional application of $A^r$ does not increase the error norm. Thus, the bound holds with $\left\lfloor \frac{t}{k+1} \right\rfloor$, completing the proof.
\end{proof}

\subsection{Proof of Proposition \ref{prop:switching_consensus}}

\begin{proof}
Since the matrix $BA^k$ is primitive, it admits a unique stationary distribution $\boldsymbol{\pi}$. Standard mixing-time analysis (see, e.g.,~\cite{bremaud2013markov, aldous-fill-2014, levin2017markov}) provides constants $C > 0$ such that
\begin{equation}
\|(BA^k)^n - \mathbf{1}\boldsymbol{\pi}^T\|_{\max} \leq C\rho_2(BA^k)^n. \label{eq:mixing_time}
\end{equation}
Since $A$ is stochastic, we have $\|A^r\|_\infty = 1$ for all integers $r \geq 1$. Therefore, for each integer $1 \leq r \leq k$, it follows that
\begin{align}
\|A^r(BA^k)^t - \mathbf{1}\boldsymbol{\pi}^T\|_{\max}
&= \|A^r(BA^k)^t - A^r\mathbf{1}\boldsymbol{\pi}^T\|_{\max} \notag \\
&= \|A^r\bigr((BA^k)^t - \mathbf{1}\boldsymbol{\pi}^T \bigr)\|_{\max}  \notag \\
&\leq \|(BA^k)^t - \mathbf{1}\boldsymbol{\pi}^T\|_{\max}.
\label{eq:stochastic_bound}
\end{align}
When $t = n(k+1)$, the product simplifies as
\[
\prod_{s=1}^{n(k+1)} M(s) = (BA^k)^n,
\]
and by applying~\eqref{eq:mixing_time}, we obtain
\begin{equation*}
\|(BA^k)^n - \mathbf{1}\boldsymbol{\pi}^T\|_{\max} \leq C\rho_2(BA^k)^n. 
\end{equation*}
For general times $t = n(k+1) + r$ with $0 \leq r \leq k$, we can decompose
\[
\prod_{s=1}^{t} M(s) = A^r(BA^k)^n.
\]
Using~\eqref{eq:stochastic_bound}, we get
\[
\|A^r(BA^k)^n - \mathbf{1}\boldsymbol{\pi}^T\|_{\max} \leq \|(BA^k)^n - \mathbf{1}\boldsymbol{\pi}^T\|_{\max} \leq C\rho_2(BA^k)^n.
\]
By Perron-Frobenius theorem, $\rho_2(BA^k) < 1$, hence the error bound decreases exponentially with $n$, ensuring
\begin{equation*}
\lim_{t\to\infty}\prod_{s=1}^{t}M(s) = \mathbf{1}\boldsymbol{\pi}^T. 
\end{equation*}
This limit implies that consensus is reached with the stationary distribution $\boldsymbol{\pi}$.
\end{proof}

\subsection{Proof of Proposition \ref{prop:pi_convergence}}
\begin{proof}
Standard mixing-time results ~\cite{bremaud2013markov, aldous-fill-2014, levin2017markov} yield
\begin{equation*}
\|A^k-\mathbf{1}\boldsymbol{\pi}^{T}\|_{\max} \le C\rho_2(A)^k
\end{equation*}
for some constant $C>0$.

Express $A^k$ as
\[
A^k=\mathbf{1}\boldsymbol{\pi}^{T}+ E_k,
\]
where $\|E_k\|_{\max} \leq C\rho_2(A)^k$.

To analyze the perturbation $BA^k - A^{k+1}$, we compute
\begin{align*}
    \|BA^k-A^{k+1}\|_{\max} &= \|(B-A)A^k\|_{\max} \nonumber \\
    &= \|(B-A)(\mathbf{1}\boldsymbol{\pi}^{T}+ E_k)\|_{\max}\nonumber \\
    &= \|(B-A)E_k\|_{\max}  \nonumber\\
    &\leq  \|E_k\|_{\max}  \nonumber \\
    &\leq C \rho_2(A)^k.
\end{align*}
By Corollary~\ref{cor:stability_bound}, this implies
\begin{align*}
    \| \boldsymbol{\pi}^k -\boldsymbol{\pi}\|_{\max} = O\!\bigl( \rho_2(A)^k\bigr)
\end{align*}
Since $x_i(0)\in[0,1]$ for all $i$, the same bound applies to the consensus values, completing the proof.
\end{proof}

\subsection{Proof of Proposition \ref{prop:perturbation_stability}}
\begin{proof}
    Regarding a small perturbation of \(A\), as \(A\) and \(B\) are stochastic matrices, we have
\begin{align*}
    \|BA^k - A^{k+1}\|_{\max} &=  \|(B-A)A^k \|_{\max} \nonumber \\& \leq  N\|(B-A) \|_{\max} \|A^k\|_{\max}. 
\end{align*}
By applying Corollary~\ref{cor:stability_bound} together with this inequality, we obtain the stated result. 
\end{proof}

\subsection{Additional Numerical Experiments}
\label{sec:numerical_results_app}


We further validate our theoretical findings through experiments on synthetic networks. These experiments demonstrate the key properties of multilayer consensus dynamics and confirm that multilayer networks can reach consensus even when individual layers fail to converge independently.

\subsubsection{Merged Layers Opinion Dynamics}

We compare Barabási–Albert (BA) networks with Erdős–Rényi (ER) networks. The BA model uses preferential attachment, which produces hub nodes with high degree, while the ER model connects each pair of nodes independently with probability \(p\). We adjust the parameters to ensure networks of size \(N=100\) and average degree 10. Initial opinions are drawn uniformly from \([0,1]\), except that the five highest-degree nodes in the BA networks are fixed at opinion 0 to model influential agents. Due to the concentration of connections around these hubs, the BA networks alone converge to a lower consensus level than the ER networks; however, the merged-layer mechanism alleviates this deficit by blending the two interaction structures. Figure~\ref{fig:merge_layer_consensus} illustrates the error convergence, consensus value, and SLEM as functions of the weighting factor $\alpha$. In particular, Figures~\ref{fig:merge_consensus_value} and \ref{fig:merge_slem} support the bounds in Propositions ~\ref{prop:consensus_bound} and ~\ref{prop:slem_bound}, respectively.

\begin{figure*}[t]
\centering
\subfloat[Error convergence for different values of \(\alpha\), where dashed lines indicate the theoretical rate and solid lines show the observed rate]{%
  \includegraphics[width=.32\textwidth]{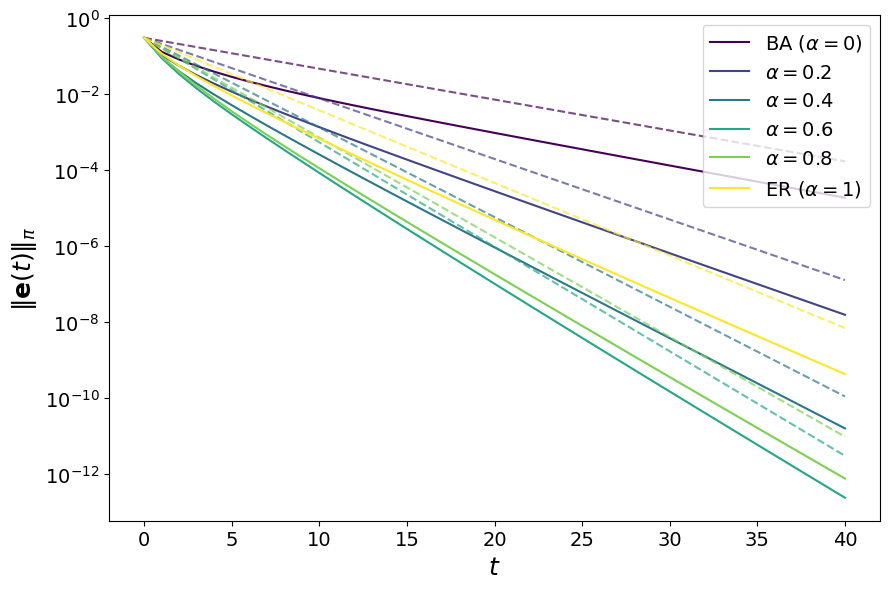}%
  \label{fig:merge_error_convergence}
}\hfill
\subfloat[Consensus value versus \(\alpha\), with red (upper bound) and green (lower bound) lines from Proposition~\ref{prop:consensus_bound}]{%
  \includegraphics[width=.32\textwidth]{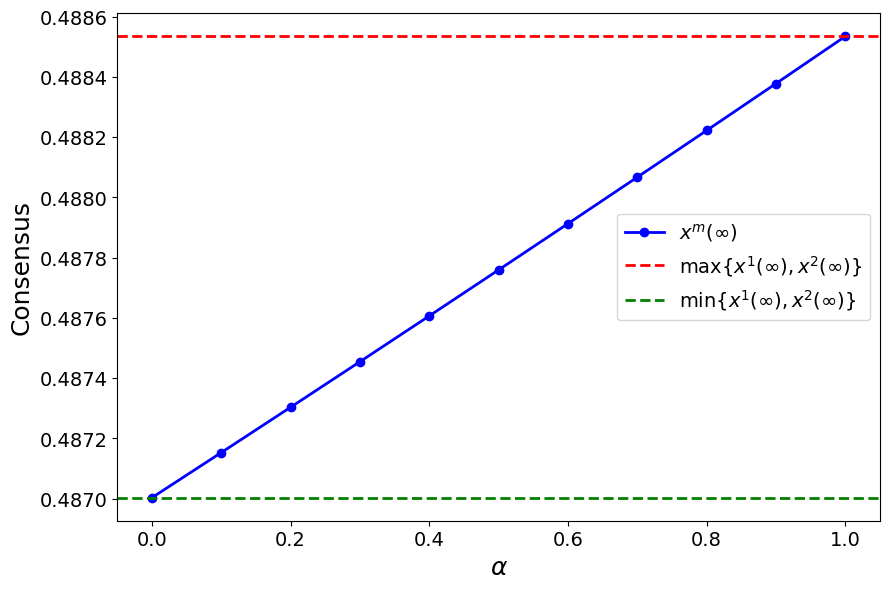}%
  \label{fig:merge_consensus_value}
}\hfill
\subfloat[SLEM versus \(\alpha\), with red (upper bound) and green (lower bound) lines from Proposition~\ref{prop:slem_bound}]{%
  \includegraphics[width=.32\textwidth]{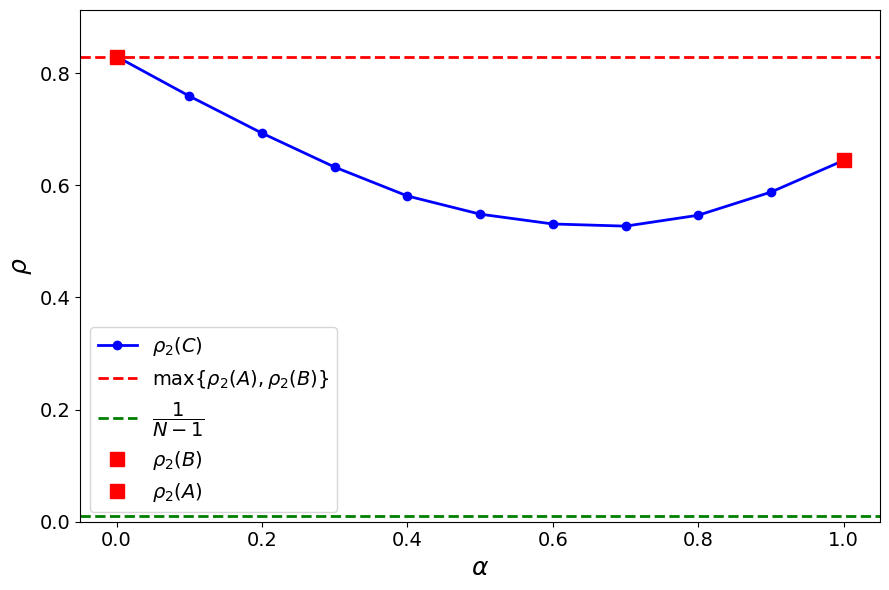}%
  \label{fig:merge_slem}
}
\caption{Merged-layer consensus experiments on random networks.}
\label{fig:merge_layer_consensus}
\end{figure*}

\subsubsection{Switching Layers Opinion Dynamics}

We construct layer \(A\) as a 6-regular graph and layer \(B\) as an 8-regular graph on \(N=100\) nodes. Figure~\ref{fig:switching_bounds} validates the convergence bound from Proposition~\ref{prop:switching_convergence}. The stepwise convergence reflects the periodic switching structure, and the overall rate depends on both \(\rho_2(B A^k)\) and the switching period \(k+1\).
\begin{figure*}[t]
\centering
\begin{subfigure}[b]{0.34\textwidth}
  \includegraphics[width=\linewidth]{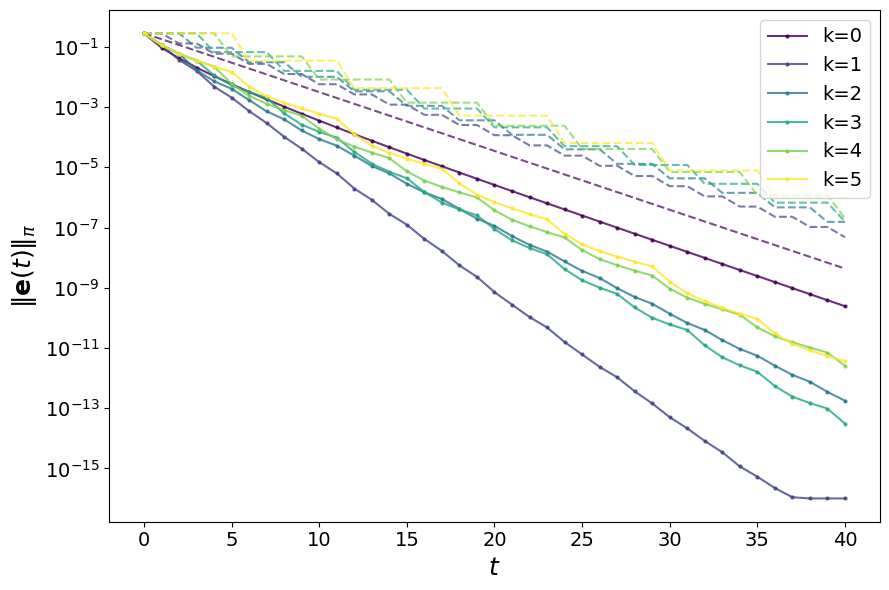}
  \caption{Error convergence for various switching periods \(k+1\), where dashed lines are theoretical and solid lines are empirical rates.}
  \label{fig:switching_bounds_a}
\end{subfigure}\hspace{1in}
\begin{subfigure}[b]{0.36\textwidth}
  \includegraphics[width=\linewidth]{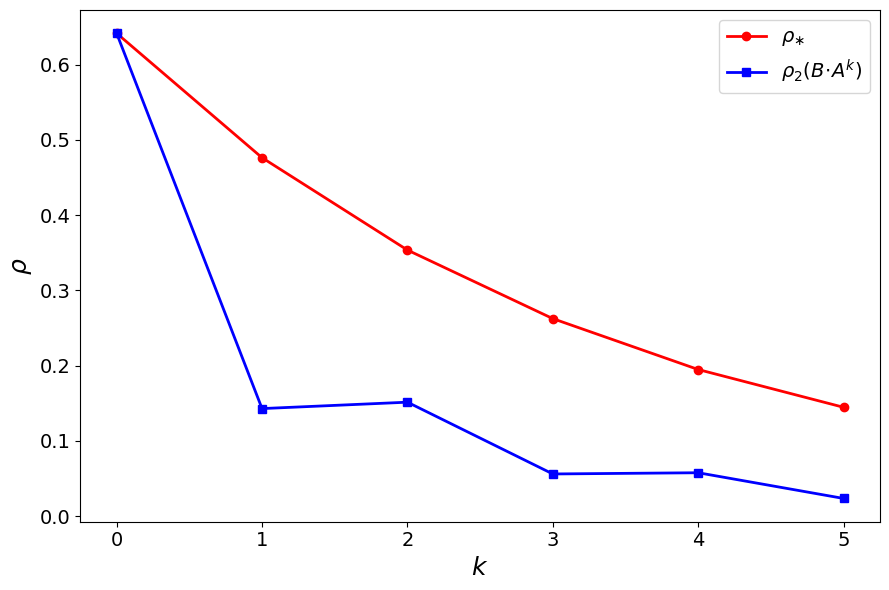}
  \caption{SLEM versus switching period \(k+1\), with upper bound from Proposition~\ref{prop:switching_convergence}}
  \label{fig:switching_bounds_b}
\end{subfigure}
\caption{Switching layer consensus experiments on random networks.}
\label{fig:switching_bounds}
\end{figure*}

\end{document}